\documentclass[12pt,reqno]{amsart}
\usepackage{amsmath,amsfonts,amsthm,amssymb,bm,bbold,dsfont,hyperref,times,enumerate}
\usepackage{parskip,xcolor}

\newtheorem{thm}{Theorem}[section]
\newtheorem*{thm*}{Theorem}
\newtheorem{lemma}[thm]{Lemma}
\newtheorem{prop}[thm]{Proposition}
\newtheorem{cor}[thm]{Corollary}

\newtheorem*{rmk*}{Remark}

\newcommand{\be}{\begin{equation}}
\newcommand{\ee}{\end{equation}}

\def\lam0{{\lambda+i0}}

\makeatletter

\newcommand{\R}{\mathbb R}
\newcommand{\Z}{\mathbb Z}
\newcommand{\E}{\mathbb E}
\newcommand{\C}{\mathbb C}
\newcommand{\pr}{\mathbb P}
\newcommand{\tr}{\mbox{tr }}
\newcommand{\abs}[1]{\left\lvert #1 \right\rvert}
\numberwithin{equation}{section}
\begin{document}

\title[Eigenvalue counting inequalities with applications]{Eigenvalue counting inequalities, with applications to  Schr\"{o}dinger operators}

\author{Alexander Elgart}
\address[A. Elgart]{Department of Mathematics; Virginia Tech; Blacksburg, VA, 24061, USA}
 \email{aelgart@vt.edu}
\author{Daniel  Schmidt}
\address[D. Schmidt]{Department of Mathematics; Virginia Tech; Blacksburg, VA, 24061, USA}
 \email{dfs@vt.edu}

\begin{abstract}
We derive a sufficient condition for a Hermitian $N\times N$ matrix $A$ to have at least $m$ eigenvalues (counting multiplicities) in the interval $(-\epsilon, \epsilon)$. This condition is expressed in terms of the existence of a principal $(N-2m)\times (N-2m) $ submatrix of $A$ whose Schur complement in $A$ has at least $m$ eigenvalues in the interval $(-K\epsilon, K\epsilon)$, with an explicit constant $K$.
  
We apply this result to a random Schr\"odinger operator $H_\omega$, obtaining a criterion that allows us to control the probability of having $m$ closely lying eigenvalues for $H_\omega$---a result known as an $m$-level Wegner estimate.  We demonstrate its usefulness by verifying the input condition of our criterion for some physical models.  These include the Anderson model and random block operators that arise in the Bogoliubov-de Gennes theory of dirty superconductors.  

\end{abstract}
\maketitle


\section{Introduction}

\subsection{Small eigenvalues and the Green function}
Let $A$ be an invertible Hermitian $N\times N$ matrix with inverse $A^{-1}$, and let $I_N$ be the $N\times N$ identity matrix. Let $G(x,y)$ denote the matrix element in the $(x,y)$ position of $A^{-1}$, also known as the Green function of $A$.   Our first  objective in this work is to relate information about the small eigenvalues of $A$ to the behavior of $G(x,y)$.  Let us denote by $C_\epsilon(A)$ the number of eigenvalues (counting multiplicities) of $A$ in the interval $I_\epsilon:=(-\epsilon,\epsilon)$. As a first step, let us ask the most basic question: Does $A$ has at least one eigenvalue in the interval $I_\epsilon$?  A well known result in the matrix analysis says that
\[C_\epsilon(A)>0 \ \Longleftrightarrow  \ \|A^{-1}\|>\frac{1}{\epsilon}.\]
Since $\|B\|_{\text{max}}\leq \|B\|\leq N\|B\|_{\text{max}}$ for any $N\times N$ matrix $B$ with
\[\|B\|_{{{\text{max}}}}=\max_{x,y}|B(x,y)|,\] we obtain the relations
\be\label{eq:basic1}C_\epsilon(A)>0  \Rightarrow   \mbox{There exists a pair } \{x,y \}\mbox{ such that }|G(x,y)|>\frac{1}{N\epsilon}\,;\ee
\[ |G(x,y)|>\frac{1}{\epsilon} ~~ \mbox{for some pair } \{x,y\} \Rightarrow  C_\epsilon(A)>0\,.\]
It is  natural to try to  quantify these relations further, viz. to detect whether the matrix $A$ has at least $m$ small eigenvalues from the behavior of $G(x,y)$. To this end, we  prove the following result.
\begin{thm}\label{thm:reduction}
Let $A=A^*$ be an $N\times N$ invertible matrix.  Let $A [\alpha,\beta]$ denote the submatrix of $A$ with rows indexed by index subset $\alpha$ and columns indexed by  index subset $\beta$.
Consider the following two assertions:
\begin{enumerate}[I.]
\item\label{i:1}  $C_{\epsilon}(A)\ge m$;
\item\label{i:2}
 There exist  index subsets \[\alpha_m=\{i_1,\ldots,i_m\},\quad  \beta_m=\{j_1,\ldots,j_m\}\] of $\{1,\ldots,N\}$  such that
\begin{equation}\label{eq:Min}
A^{-1}[\alpha_m,\beta_m]A^{-1}[\beta_m,\alpha_m]> \frac{K^2}{\epsilon^2}\,I_N[\alpha_m,\alpha_m]\mbox{\ \  for some }K>0,
\end{equation}
where $I_N$ is the $N\times N$ identity matrix.
\end{enumerate}
Then  \eqref{i:1}  implies \eqref{i:2} with
\be\label{eq:K_2}K=
\frac{C_m}{N},\quad C_m=\frac{1}{m! \, 2^{m-1} }\ee
%

Conversely,
 \eqref{i:2} with $K=1$  implies \eqref{i:1}.
\end{thm}
The constant $C_m$  in \eqref{eq:K_2} is not sharp for $m>1$. However, the dependence on $N$ is optimal (and we will be interested in small $m$,  large $N$ behavior in the application below).

It is often convenient to work with  principal submatrices $A^{-1}[\gamma]$ of $A^{-1}$. One can tailor  Theorem \ref{thm:reduction} somewhat differently to accommodate this requirement, at the cost of increasing the cardinality of the corresponding index subsets $\alpha_m$, $\beta_m$. Namely, we have the following result:

\begin{cor}\label{cor}
Let $A=A^*$ be an $N\times N$ invertible matrix. Consider the following two assertions:
\begin{enumerate}[I.]
\item\label{i:1'}  $C_{\epsilon}(A)\ge m$;
\item\label{i:2'}
 There exists an  index subset $\gamma_m=\{i_1,\ldots,i_{2m}\}$ of $\{1,\ldots,N\}$  such that for any subset $\gamma\supset \gamma_m$ for which the matrix $A^{-1}[\gamma]$ is invertible
\begin{equation}\label{eq:Min'}
C_{\epsilon/K}\left(\left(A^{-1}[\gamma]\right)^{-1}\right)\ge m\mbox{\ \  for some }K>0.
\end{equation}
\end{enumerate}
Then  \eqref{i:1'}  implies \eqref{i:2'} with
\[K=
\frac{C_m}{N},\quad C_m=\frac{1}{m! \, 2^{m-1} }\]
%

Conversely,
 \eqref{i:2'} with $K=1$  implies \eqref{i:1'}.
\end{cor}
\begin{rmk*}
The matrix $\left(A^{-1}[\gamma]\right)^{-1}$ coincides with the Schur complement of $A[\gamma^c]$ in $A$, see \eqref{eq:Schr} below for details. Here $\gamma^c = \{1, 2,\ldots, N\}\setminus\gamma$.
\end{rmk*}

\subsection{Application to random Schr\"odinger operators}
In  quantum physics, the tight-binding approximation is often used as the prototypical model for the study of electron propagation in solids.  In this model, the evolution of the wave function $\psi$ on the $d$-dimensional lattice $\Z^d$ is given by the Schr\"odinger equation
\be\label{eq:tight} i \hbar \dot{\psi}_t = H \psi_t;\quad \psi(0)=\psi_0,\ee
where the self-adjoint Hamiltonian $H$ is a sum of the hopping term $H_0$ and the potential $V$, of the form
\[
(H\psi)(x)=(H_0\psi)(x)+ V(x)\psi(x),\quad x\in\Z^d.
\]
In this work we consider the random operators that have this functional form. Let us list few of these:
\vspace{3mm}
\paragraph{\it Anderson model $H_A$}
One of the best-studied models for disordered solids was introduced by P. H. Anderson in \cite{An}.  In this model the Hilbert space is $\ell^2(\Z^d)$, the hopping term $H_0$ is the discrete Laplacian $\Delta$, and  the potential $V$ in $H$ above is of the form $V(x)=g \sum a_{x-y}v(y)$, where the single site potentials $v(y)$ are independent random variables. The real parameter $g$ is a coupling constant which describes the strength of the disorder.
\vspace{3mm}
\paragraph{\it Alloy-type Anderson model $H_{alloy}$}
Here the Hilbert space is  also $\ell^2(\Z^d)$, the hopping term $H_0$ is a short range ergodic operator. The value of the potential
$V(x)$ at a site $x\in \mathbb{Z}^d$ is generated from independent random variables $\{u(y)\}$ via the transformation
\[ V(x) = g\sum_{y\in\mathbf{\Gamma}} a_{x-y} u(y)~,\]
where the index $y$ takes values in some sub-lattice $\mathbf{\Gamma}$ of $\mathbb{Z}^d$.  The Hamiltonian $H_A$ coincides with $H_{alloy}$ provided $H_0=\Delta$, $\mathbf{\Gamma}=\Z^d$; $a_z=\delta_{|z|}$, where $\delta_x$ is Kronecker delta function: $\delta_0=1;\,\delta_x=0$ for $x\neq0$. In general, the random potential at sites $x,y$ is correlated for this model. As its name suggest, $H_{alloy}$ is used to describe (random) alloys in the tight binding approximation.
\vspace{3mm}
\paragraph{\it Random block operator $H_{block}$}
The Hilbert space is  $\ell^2(\Z^d;\C^k)\cong \left(\ell^2(\Z^d)\right)^k$ (the space of square-summable functions $\psi\hspace{-1mm}: \Z^d\rightarrow \C^k$).  The kernel $H_0(x,y)$ of the hopping term  is a deterministic, translation invariant $k\times k$ matrix. The random potential $V(x)$ at each site is an independently drawn random $k\times k$ Hermitian matrix multiplied by $g$.

\subsubsection{Previous Results}


Anderson  \cite{An} argued  that in the $g\gg1$ regime, the solution of the initial value problem \eqref{eq:tight} for $H_A$ stays localized in space for all times almost surely if the initial wave packet $\psi_0$ is localized. Mathematical study of Anderson localization is an active field; we refer readers to the recent reviews \cite{Kirsch,S} on the subject for the detailed bibliography. In this work, we focus our attention on a single aspect of Anderson localization---the so-called $m$-level Wegner estimate.


Let  $|S|$ denote  the cardinality of the set $S$.
Let $H_A^\Lambda$ be a  restriction of the operator $H_A$ to a finite box $\Lambda$. Then the $m$-level Wegner estimate is an upper bound on the probability of $n$ eigenvalues being in the same energy interval $I_\epsilon:=(E-\epsilon,E+\epsilon)$:
 \[
 \pr\left({\mathcal C}_\epsilon(H_A^\Lambda-E)\ge m\right)\le C_m(|\Lambda|\epsilon)^m,
 \]
for random variables $v(x)$ with a bounded density. As such, it  gives some measure of the correlation between multiple eigenvalues. We will refer to the $1$-level bound  simply as the  Wegner estimate (first established by F. J. Wegner in \cite{W}).  It plays the instrumental role in the proof of Anderson localization.

If localization occurs in some energy interval $I\subset\R$, the entire spectrum of $H_A$ in $I$ is pure point. It is then natural to study the distribution of the eigenvalues for $H_A^\Lambda$ in this interval.  Physicists expect that there is no energy level repulsion for states in the localized regime: that is, the eigenvalues should be distributed independently on the interval $I$. The first rigorous result in this direction, namely that the point process associated with the (rescaled) eigenvalues converges to a Poisson process, was obtained by Molchanov \cite{Mo} in the setting of a one-dimensional continuum.

Minami \cite{Min} established the analogous result for $H_A$ under the assumption that the distribution
$μ$ of every $v(n)$ has a bounded density. The key component in \cite{Min} is the $2$-level Wegner estimate, which is consequently known as the Minami estimate.

By now the localization phenomenon for the original Anedrson model $H_A$ is well understood. In particular, the general $m$-level Wegner estimate is known to hold for essentially all distributions $\mu$ of the random potential $v(x)$; see \cite{BHS, GV}. We refer the reader to \cite{CGK} for the state of the art results concerning eigenvalue counting inequalities for $H_A$. However, the current understanding of many (in fact almost all) other random models of interest remains partial at best.

The Wegner estimate for a special class of alloy-type Anderson model $H_{alloy}$ was first established by Kirsch in \cite{K1}.  By now it is known to hold in fair generality (albeit not universally).  See the recent preprint \cite{LPTV} for the extensive bibliography on the subject. The Wegner estimate for  a random block operator $H_{block}$---with $V(x)=gv(x) A$ where $v(x)$ are independent random variables and $A$ is a fixed invertible Hermitian matrix---holds in perturbative regimes.  That is, it holds near the edges of the spectrum, \cite{CE} and in the strong disorder regime $1\ll g$, \cite{ESS}.  A weaker bound (weaker in terms of the volume dependence) near the edges of the spectrum was established for Fr\"ohlich model, where the matrix-valued potential is given by $V(x) = gU(x)^* A U(x)$,  where $A$ is a fixed self-adjoint $k \times k$ matrix, and the $U(x)$ are independently chosen according to the Haar measure on $SU(k)$, \cite{B}.

On the other hand, not much progress has been made on extensions of the multi-level Wegner estimate, besides allowing for more general background operators $H_0$ than the discrete Laplacian. In particular, apart from two special examples below, all previous works require a non-correlated random potential. In \cite{K}, this limitation was partially removed in the continuum one-dimensional setting, allowing for positively correlated randomness. In \cite{CGK'}, the authors announced the establishment of the Minami estimate and subsequently Poisson statistics for a general class of positively correlated random potentials. Unfortunately, although \cite{CGK'} contains a new elegant and efficient proof of Minami's estimate for $H_A$, its extension to the generalized setting has a significant gap, which so far has not been removed. Finally, let us mention the recent result \cite{TV}, which established the Minami estimate for a special class of weakly correlated randomness for which one can transform the problem to the uncorrelated one.

The reader may wonder about the glaring disparity between the wealth of results on the $1$-level Wegner estimate and the scarcity of results for its many-level counterparts. The reason can be traced to the direct (and frequently exploited) link between the former and the underlying Green function given by \eqref{eq:basic1}. The amenable nature of the Green function then allows one to establish a robust $1$-level Wegner estimate in many situations of interest. In the present work, we harness  the connection between the many-level Wegner estimate and the Green function given by Theorem \ref{thm:reduction} to establish an $m$-level Wegner estimate for a certain class of models with correlated randomness. Roughly speaking, our method works if the randomness in the system is sufficiently rich.  (We will quantify this statement in the sequel.)

Although in most known applications (such as localization, simplicity of the spectrum, and Poisson statistics of eigenvalues) one is interested in the $1$- and $2$-level Wegner estimates, it is nonetheless natural from a mathematical perspective to investigate the general many-level case. From a  practical perspective, it can yield some insight on the nonlinear Anderson model via multi-state resonance phenomenon, \cite{FKS}.

The blessing and the curse of the existing methods employed in proof of the Minami estimate (with the single exception of \cite{K}) is that the nature of the background operator $H_0$ plays little if any role in the proofs. It is however clear that in the case of the correlated random potential in $H_A$ one cannot hope to get the Minami estimate without exploiting the structure of $H_0$. Indeed, consider the one dimensional operator $H_{alloy}$ with  $H_0=0$, and the random potential at odd sites being i.i.d. random variables, while $v(2n)=v(2n-1)$. Its spectrum consists of (the closure of) the set of  eigenvalues $\{\lambda_n\}=\{v(n)\}$, each one being degenerate. Consequently, even though $H_{alloy}$ in this setting is perfectly localized, the probability of finding two closely lying eigenvalues is equal to $1$.

\subsubsection{The $m$-level Wegner estimate for the random block model $H_{block}$}
We will consider the class of random block models $H_{block}$ introduced earlier.

Let ${\mathcal G}=({\mathcal V},{\mathcal E})$ be a graph with degree at most $\kappa$, such that the set of vertices (sites) ${\mathcal V}$  is  finite with cardinality $N$. The main example of this model is the restriction of the lattice $\Z^d$ to the box, but the greater generality does not require additional effort here.

Let $(\Omega, \pr)$ be a probability space. Let $\{V(x)=gA_\omega(x):\ x\in {\mathcal V},\, \omega\in\Omega\}$ be a collection of independent, identically distributed random $k\times k$ Hermitian matrices.
\paragraph{\it Basic Assumption}
We now state the main technical condition that we will use as an input for our application theorem below.

\vspace{3mm}
\paragraph{{\bf (A)}}
{\it For an integer $n$, let $S$ be a given set of $2nk$ district integers. Then the matrix $A_\omega(x) -a$ is invertible for all $a\in S$ and all $\omega\in\Omega$.  Moreover, there exists an $\alpha>0$ such that, for any integer $a\in S$, any $\epsilon\in [0,1]$ and arbitrary Hermitian $k\times k$ matrix $J$ the bound
 \be\label{eq:assump}
\pr\left(\abs{\det\left((A_\omega(x) -a)^{-1}+(J+a)^{-1}\right)}\le \epsilon\right)\le K \epsilon^\alpha
\ee
holds.}
\vspace{3mm}

It guarantees that the randomness in the system is rich enough to imply the result below (Theorem \ref{thm:main}).  At the first glance, a more natural condition should concern the properties of the matrix $ A_\omega(x)+J$ as it is the correct functional form of the corresponding Schur complement of $H_\omega$ (see Section \ref{sec:main} for details).  However, this turns out to be an unsuitable choice because of the absence of  an $\it{a \, priori}$ bound on the norm of the background operator $J$ (which encodes the information about the environment of the $x$-block in $H_\omega$). On the other hand, for a sufficiently large set of numbers $\{a_i\}$ one can ensure that regardless of the norm of  $J$, one of  the matrices $\{(J+a_i)^{-1}\} $ is bounded in norm by $1$ (see Proposition \ref{lem:prel5}). We then exploit the fact that  matrices $(A-a)^{-1}+(J+a)^{-1} $ (which appears in \eqref{eq:assump}) and  $A+J$ are  related:
\begin{multline}\label{eq:Wood'}\left(A+J \right)^{-1} =  (A-a)^{-1} \\  - (A-a)^{-1}\left((A-a)^{-1}+(J+a)^{-1} \right)^{-1} (A-a)^{-1},\end{multline}
provided that $A-a$ and $J+a$ are invertible. One can readily verify \eqref{eq:Wood'} (which is in fact a particular case of Woodbury's matrix identity) by multiplying both sides by $A+J$.

\vspace{3mm}

We now  introduce our single-particle Hamiltonian. Namely,
let $H_\omega(g)$ be a random block operator $H_{block}$ acting on $\ell^2({\mathcal V};\C^k)$ (the space of square-summable functions $\psi:\ {\mathcal V}\rightarrow \C^k$) as
\be\label{eq:H}
(H_\omega(g)\psi)(x)=(H_0\psi)(x)+gA_\omega(x)\psi(x),
\ee
where $g>0$ is a coupling constant, $H_0$ is an arbitrary deterministic self-adjoint operator on  $\ell^2({\mathcal V};\C^k)$, and $A_\omega(x)$ is an independently drawn random $k\times k$ Hermitian matrix as above. We use the notation $H_\omega(g)$ instead of $H_{block}$ to stress the random nature of this operator as well as the dependence on the parameter $g$.

\begin{thm}\label{thm:main}
Assume {\bf (A)}. Then
\begin{enumerate}[I.]
\item
 For any $E\in\R$ the operator $H_\omega(g)-E$ is almost surely invertible.
\item
Moreover, there  exist $\epsilon_0>0$ and $C>0 \,$ (which depend only on $k,m,\alpha$) for which we have
 \be\label{eq:mainbnd}
 \pr\left({\mathcal C}_\epsilon(H_\omega(g)-E)\ge m\right)\le C \left|\ln(N\epsilon/g)(N\epsilon/g)^\alpha\right|^{m}
 \ee
 for any $E\in\R$, for any  $\epsilon\in[0,\epsilon_0]$ and for all $m\le n$. In the $m=1$ case we can improve the above bound to
  \be\label{eq:mainbnd1}
 \pr\left({\mathcal C}_\epsilon(H_g-E)\ge 1\right)\le C (N\epsilon/g)^\alpha.
 \ee
 \end{enumerate}
 \end{thm}
 \subsubsection{Examples}

\begin{enumerate}
\item[]
\end{enumerate}
\vspace{-5mm}
$\bullet$ Anderson model $H_A$. 
As we mentioned earlier, the nontrivial Minami estimate is well understood only for the original Anderson model among all alloy-type models. It is therefore a litmus test to verify Assumption {\bf (A)} for $H_A$.
\begin{thm}\label{thm:weakMint}
Suppose that the distribution $\mu$ of the $v(x)$ variables in $H_A$ is compactly supported  on the interval $I=[-b,b]$ for some $b>0$ and is $\beta$-regular, i.e. for any Lebesgue - measurable $S\subset I$ we have
\[ \mu(S)\le C|S|^\beta.\]
Then Assumption {\bf (A)} holds with $\alpha=\beta$.
\end{thm}
Our approach to the Minami estimate is also meaningful for the $\mathbf{\Gamma}$-trimmed Anderson model introduced in \cite{EK}, near the edges of the spectrum, in the sense that  the assumption {\bf (A)} can be  verified for it.

$\bullet$ Fr\"ohlich model and alloy type Anderson model $H_{alloy}$.  Assumption {\bf (A)} is either not satisfied for a single site $x$ or is satisfied with a power $\alpha$ which is too small to make the result meaningful. However, the close inspection of the proof of Theorem \ref{thm:main} shows that the matrix $J$ that appears in Assumption {\bf (A)} is not required to be  completely arbitrary. In fact, the relevant matrices $J$ carry the structure of the Schr\"odinger operator (with arbitrary boundary conditions). It seems plausible  (and is on our to-do list) that Assumption {\bf (A)} can be verified for such $J$ and sets of sites that include $x$ and its neighbors. 

%
%
%
%

$\bullet$ The third model arises from the study of dirty superconductors via the Bogoliubov - de Gennes equation. After a suitable change of the coordinate basis, the Bogoliubov-de Gennes (BdG) model above can be described in terms of the operator defined in \eqref{eq:H}, with  $A_\omega(x)=\sigma_\omega(x)\in M_{2\times2}$, where $\sigma_\omega(x)$ is a random Pauli matrix of the form
\be\label{eq:sgm1}\sigma_\omega(x)=\left[\begin{array}{cc}
u_x & v_x \\
 v_x &-u_x\\
\end{array}\right];\quad u_x,\, v_x \mbox{ are  random variables.}\ee
The $m$-Wegner estimate for these models has only been established for $m=1$ case, for a restricted class of joint distributions of $u,\,v$  variables (absolutely continuous and with support in a half plane), and for a  specific background operator $H_0$ in \cite{KMM,GM}. 
We  establish the robust Wegner and Minami estimates for this model.
\begin{thm}\label{thm:weakMin}
Let each $A_\omega(x)$ be given by \eqref{eq:sgm1}. Suppose that the joint distribution $\mu$ of the $u,v$ variables is supported on a unit disc $O$ and is $\beta$-regular, i.e. for any Lebesgue - measurable $S\subset O$ we have
\[ \mu(S)\le C|S|^\beta.\]
Then Assumption {\bf (A)} holds with $\alpha=\beta$.
\end{thm}

 \subsection{Paper's organization}
We prove our main abstract result, Theorem \ref{thm:reduction}, along with its corollary, in Section \ref{sec:reduce}.  We prove our  result on eigenvalue estimates, Theorem \ref{thm:main}, in Section \ref{sec:main}.  We consider the implication of the latter result for the random block operators in Section \ref{sec:appl}.  These proofs depend on a number of auxiliary results, which we prove in Section \ref{sec:proofs}.


\section{Proof of Theorem \ref{thm:reduction} and Corollary \ref{cor}}\label{sec:reduce}
\subsection{Notation}
Let $n$ be a positive integer, and let $\alpha$ and $\beta$ be index sets, i.e., subsets of $\{1, 2,\ldots, n\}$. We denote the cardinality of an index set by $|\alpha|$ and its complement by $\alpha^c = \{1, 2,\ldots, n\}\setminus\alpha$. For an  $n\times n$ matrix $A$,  let $A [\alpha,\beta]$ denote the submatrix of $A$ with rows indexed by $\alpha$ and columns indexed by $\beta$, both of which are thought of as increasing, ordered sequences, so that the rows and columns of the submatrix appear in their natural order.
We will write $A[\alpha]$ for $A[\alpha,\alpha]$. If $|\alpha|=|\beta|$ and if $A[\alpha,\beta]$ is nonsingular, we denote by $A/A[\alpha,\beta]$ the Schur complement of $A[\alpha,\beta]$ in $A$, \cite{Z}:
\be\label{eq:Schr}
A/A[\alpha,\beta]=A[\alpha^c,\beta^c]-A[\alpha^c,\beta]\left(A[\alpha,\beta]\right)^{-1}A[\alpha,\beta^c].
\ee
We will frequently use Schur's complementation and its consequences in this work; we refer the reader to the comprehensive book \cite{Z} on this topic.

For a Hermitian matrix $A$ and a positive number $a$ we will write
\[{\mathcal B}_a(A):=|\sigma(A)\cap[a,\infty)|.\]

Let $P_\epsilon(A)$ denote the spectral projection of the Hermitian matrix $A$ onto the interval $(-\epsilon,\epsilon)$ for $\epsilon>0$.
\subsection{Proof of Theorem \ref{thm:reduction}}

Suppose that (\ref{i:1}) holds. We will use the following assertion:
\begin{prop}\label{thm:aux}
Let $A$ be an $N \times N$ positive definite matrix, and suppose that ${\mathcal B}_a(A)=k$ for some $a>0$.  Then there exists an index subset $\alpha_k=\{i_1, i_2,\ldots, i_k\}$ of $\{1,\ldots,N\}$ such that  $A[\alpha_k]\ge \frac{a}{k! \, 2^{k-1}\,N}I_N[\alpha_k]$ .
\end{prop}

 By Proposition  \ref{thm:aux} there exists $\alpha_m$ such that
\[P_\epsilon(A)[\alpha_m]\ge\frac{C_m}{ N} I_N[\alpha_m].\]
with $C_m = \frac{1}{k! \, 2^{k-1}}$.  Combining this bound with
\[A^{-2}> \frac{1}{\epsilon^2}P_\epsilon(A)\]
we obtain
\[A^{-2}[\alpha_m] > \frac{C_m}{ N\epsilon^2}\,I_N[\alpha_m].\]
Since $\sigma(TT^*)\setminus\{0\}=\sigma(T^*T)\setminus\{0\}$ for any operator $T$, we deduce from the previous equation (with $T=I_N[\alpha_m,\alpha_N] A^{-1}$) that there exists an orthogonal projection $Q$ of rank $m$ such that \[A^{-1}I_N[\alpha_m]A^{-1} > \frac{C_m}{N\epsilon^2}\,Q.\]
Applying now Proposition  \ref{thm:aux} once again, we conclude that there exists  $\beta_m$ such that \eqref{eq:Min}  holds with $K$ given by \eqref{eq:K_2}.

\vspace{.3cm}
\noindent
Conversely, suppose that \eqref{eq:Min} holds with $K=1$. Since
\[A^{-2}[\alpha_m]\ge
A^{-1}[\alpha_m,\beta_m]A^{-1}[\beta_m,\alpha_m],\]
the assertion follows from the Cauchy interlacing theorem for the Hermitian matrix $A^{-2}$ and its principal submatrix $A^{-2}[\alpha_m]$.
\hfill \qed



\begin{proof}[Proof of Proposition \ref{thm:aux}]
The proof will proceed by induction in $k$. If $k=1$, the result follows from the fact that $A$ is positive, so $\tr A=\sum_{\lambda\in\sigma(A)}\lambda\ge a$. Since the trace is at least $a$, there exists a diagonal entry which is greater than or equal to $\frac{a}{N}$.

Suppose we have established the induction hypothesis for $k=K$.  We want to verify the induction step, i.e. the case $k=K+1$. To this end, choose the index $i_1$ so that $A_{i_1i_1}\ge A_{ii}$ for all $i$. Without loss of generality, let us assume that $i_1=1$. Then $A$ is of the block form
\[A=\left[\begin{array}{cc}A_{11}&u\\u^*& B\end{array}\right].\]
Consider now the matrix $D=A/A_{11}$. It is positive definite by the Schur complement condition for positive definiteness (as $A$ is positive definite). Also, the matrix
\[F=\left[\begin{array}{cc}A_{11}&u\\u^*& B-D\end{array}\right]\]
is rank one (since $F/A_{11}=0$), so by the rank one perturbation theory, ${\mathcal B}_a(A-F)\ge K$. But ${\mathcal B}_a(A-F)={\mathcal B}_a(D)$. Using the induction hypothesis, we conclude that there exists an index set $\alpha_K=\{i_2,\ldots,i_K\}$ with $1\notin \alpha_K$ such that
$D[\alpha_K]\ge \frac{a}{K!\,2^{K-1}N}\,I_N[\alpha_K]$.

The induction step (with $\alpha_{K+1}=\alpha_K\cup\{1\}$) now follows from the following assertion:


\begin{lemma}\label{lem:st1}
Let $A$ be an $l\times l$ positive definite matrix of the block form
\be\label{eq:conda}A=\left[\begin{array}{cc}A_{11}&u\\u^*& B\end{array}\right].\ee
Suppose that in addition $A_{11}\ge A_{ii}$ for all $i\in\{1,\ldots,l\}$, and $A/A_{11}\ge a$ for some $a>0$. Then $A\ge \frac{a}{2l}$.
\end{lemma}

\end{proof}


\begin{proof}[Proof of Lemma \ref{lem:st1}]
To show that $A- \frac{a}{2l}\ge0$ it suffices to check (by the Schur complement condition for positive definiteness) that
\be\label{eq:bndsi}
A_{11}- \frac{a}{2l}\ge 0;\quad(A- \frac{a}{2l})/(A_{11}- \frac{a}{2l})\ge 0.\ee
Since $A/A_{11}\ge a$, we have $A_{ii}\ge a$ for all $i\ge2$ as $a$ is positive, so by assumption of the lemma $A_{11}\ge a$ as well (and hence we have established the first bound in \eqref{eq:bndsi}). Next we write
\begin{eqnarray}&& (A- \frac{a}{2l})/(A_{11}- \frac{a}{2l})=B- \frac{a}{2l}-\frac{u^*u}{A_{11}- \frac{a}{2l}}\nonumber \\ &=&\left(B-\frac{u^*u}{A_{11}}\right)- \frac{a}{2l}- \frac{au^*u}{2lA_{11}(A_{11}- a/2l)}\label{eq:siin2}\\ &\ge& a-\frac{a}{2l}- a\frac{u^*u}{l(A_{11})^2}.\nonumber
\end{eqnarray}
Now observe that since $A$ is positive, the contraction $A[\{1,i+1\}]$ is also positive for all $i$, and in particular $\det  A[\{1,i+1\}]=A_{11}B_{ii}-|u_i|^2\ge0$. But $A_{11}\ge B_{ii}$ for all $i$, hence $|u_i|^2/(A_{11})^2\le 1$. We therefore can estimate
\[\left\|u^*u\right\|=\|u\|^2\le (l-1)(A_{11})^2.\]
Substitution of this estimate into \eqref{eq:siin2} yields the second bound in  \eqref{eq:bndsi}.
\end{proof}

\subsection{Proof of Corollary \ref{cor}}
We first observe that if sets of indices $\alpha,\beta$ satisfy $\alpha\subset\beta$, then  $A^{-1}[\alpha]$ is a principal submatrix of $A^{-1}[\beta]$, and we have
\be\label{eq:Cauint}C_\epsilon\left(\left(A^{-1}[\beta]\right)^{-1}\right)\ge C_\epsilon\left(\left(A^{-1}[\alpha]\right)^{-1}\right)\ee
by the Cauchy interlacing theorem, provided the matrices $A^{-1}[\alpha], A^{-1}[\beta]$ are invertible. Therefore, it suffices to establish the corollary for the smallest set $\gamma_{min}$ that contains $\gamma_m$ and for which $A^{-1}[\gamma_{min}]$ is invertible. Without loss of generality we will assume that $\gamma_{min}=\gamma_m$.

\vspace{3mm}

Suppose that \eqref{i:1} holds. Then the assertion \eqref{i:2} of Theorem \ref{thm:reduction} holds with $K$ given by \eqref{eq:K_2}. Construct now the set $\gamma_m=\alpha_m\cup \beta_m$ with $\alpha_m$, $\beta_m$ from the assertion \eqref{i:2} of Theorem \ref{thm:reduction}. Let us consider the matrix \[B:=\left(A^{-1}[\gamma_m]\right)^{-1}.\]
Since $A^{-1}[\alpha_m,\beta_m]$ is a submatrix of  $A^{-1}[\gamma_m]$ we see that the condition \eqref{i:2} of Theorem \ref{thm:reduction} is fulfilled for $B$. Hence we can apply Theorem \ref{thm:reduction}  to $B$ to conclude that $C_{\epsilon/K}(B)\ge m$.

\vspace{3mm}

Conversely, suppose that \eqref{eq:Min'} holds with $K=1$. Then  \eqref{i:1} holds as well, as follows from \eqref{eq:Cauint} with $\alpha=\gamma_m$, $\beta=\{1,\ldots,N\}$.
\hfill \qed


 \section{Proof of Theorem \ref{thm:main}}\label{sec:main}

 We first observe that by scaling it suffices to prove the result for the $g=1$ case. We will use the shorthand notation $H_\omega$ instead of $H_\omega(1)$ in the sequel.

 Next, we prove the first assertion of the theorem, using induction in $N$.  To initiate the induction, we consider the case $N=1$, so that $H_\omega=A_\omega(x)+K$, where $K$ is a deterministic Hermitian matrix. It follows from Assumption  {\bf (A)} and \eqref{eq:Wood'} that $H_\omega-E$ is invertible almost surely.

Suppose now that the induction hypothesis holds, i.e. the matrix $H_\omega-E$ is almost surely invertible for $N\le M$ and all $E$. We want to establish the induction step ($N=M+1$ case). To this end, let $\hat {\mathcal V} $ be any subset of ${\mathcal V}$ of cardinality $M$, and let $\hat H_\omega$ be a restriction of $H_\omega$ to $\hat {\mathcal V} $. By the induction hypothesis, $\hat H_\omega-E$ is invertible almost surely for all $E$. Let us consider some configuration $\omega$ for which $\hat H_\omega-E$ is invertible. Then $ H_\omega-E$ is invertible if and only if the Schur complement of $\hat H_\omega-E$ in $ H_\omega-E$, i.e. $(H_\omega-E)/(\hat H_\omega-E)$, is  invertible, \cite{Z}. But $(H_\omega-E)/(\hat H_\omega-E)$ is a Hermitian $k\times k$ matrix of the form $A_\omega(x)+J$, where $\{x\}={\mathcal V}\setminus \hat {\mathcal V} $, and $J$ is a matrix independent of the randomness in $A_\omega(x)$. It follows by the same  argument  as in the $N=1$ case that $(H_\omega-E)/(\hat H_\omega-E)$ is invertible for almost all values of the randomness in $A_\omega(x)$.

\vspace{3mm}

 We now prove the second assertion of the theorem. We will only consider configurations $\omega$ in $\Omega$ such that $ H_\omega-E$ is invertible  (for the remaining set of configurations has measure zero by the first assertion).

 For the random operator $T_\omega$, let ${\mathcal E}_\epsilon(T_\omega)$ be the event $\{\omega:\ {\mathcal C}_\epsilon(T_\omega)\ge m\}$.  With this notation, we wish to estimate the size of the set ${\mathcal E}_\epsilon(H_\omega-E)$. If we enumerate the vertices $v\in{\mathcal V}$, we can think of $H_\omega$ as a $kN\times kN$ Hermitian matrix with a block form, i.e. the indices $\{lk-k+1,\ldots,lk\}$ correspond to the vertex $l$ in $\mathcal V$, with $l=1,\ldots,N$.

 \paragraph{\it Size reduction} We first reduce the dimensionality of the original problem using Corollary \ref{cor}. This assertion gives us the existence of the index subset $\gamma_m$ with $|\gamma_m|=2m$ such that  inclusion
 \be\label{eq:incle}{\mathcal E}_\epsilon(H_\omega-E)\,\subset \,{\mathcal E}_{\epsilon/K}\left(\left((H_\omega-E)^{-1}[\gamma]\right)^{-1}\right)\ee
holds for any index set $\gamma\supset \gamma_m$ for which $(H_\omega-E)^{-1}[\gamma]$ is invertible, with $K$ given by  \eqref{eq:K_2}.  (To be precise, the matrix size $N$ in that corollary gets replaced by $kN$.)

In general, the submatrix $(H_\omega-E)^{-1}[\gamma]$ can be a complicated object, so  it is not immediately clear that such a reduction is helpful. However, if the set $\gamma$ happens to consist of the indices that agree with the block structure of $H_\omega$, something interesting happen.  More precisely, suppose that $i\in \gamma \Rightarrow \left ( j\in\gamma \mbox{ for any } j \mbox{ with } \lfloor j/k \rfloor=\lfloor i/k \rfloor \right)$, where $\lfloor \, \cdot \, \rfloor$ is the floor function. In this case we can associate $\gamma$ with a subset ${\mathcal V}'$ of the original vertex set $\mathcal V$. Then the submatrix $\left((H_\omega-E)^{-1}[\gamma]\right)^{-1}$ retains the same block form as $H_\omega$, in the following sense: If we go back to the vertex representation for $\left((H_\omega-E)^{-1}[\gamma]\right)^{-1}$ (which is possible due to the special form of the set $\gamma$), then for any $\psi\in \ell^2({\mathcal V}'; \C^k)$ and any $x\in {\mathcal V}'$ we have
\be\label{eq:H'}
\left(\left((H_\omega-E)^{-1}[\gamma]\right)^{-1}\psi\right)(x)=(T_0\psi)(x)+A_\omega(x)\psi(x).
\ee
This can be seen from the fact that the matrix $\left((H_\omega-E)^{-1}[\gamma]\right)^{-1}$ coincides with the Schur complement of $(H_\omega-E)[\gamma^c]$ in $H_\omega-E$,
\[\left((H_\omega-E)^{-1}[\gamma]\right)^{-1}= (H_\omega-E)/(H_\omega-E)[\gamma^c].\]
It is important to note that the operator $T_0$ in \eqref{eq:H'} is independent of the randomness associated with matrices $\{A_\omega(x)\}_{x\in {\mathcal V}'}$ (though it does depend on the other random variables). We also note that the matrix $\left((H_\omega-E)^{-1}[\gamma]\right)$ is almost surely invertible (as follows from the first part of the theorem).

Combining these observations, we conclude that it is beneficial (and sufficient) to consider the sets $\gamma$ in \eqref{eq:incle} that respect the block structure of $H_\omega$ and  therefore contain up to $2km$ indices (i.e. up to $2m$ vertices in ${\mathcal V}'$, as in Corollary \ref{cor}, and exactly $k$ indices per vertex, to preserve blocks).  Thus, we have obtained the following intermediate result
\begin{lemma}\label{lem:interm}
Suppose that the second assertion of Theorem \ref{thm:main} holds for all $N\le 2m$. Then it holds for any $N$.
\end{lemma}
 \paragraph{\it Norm reduction} The deterministic part of $H_\omega$---namely the operator $H_0$---can be arbitrary large in norm (even if $\|H_0\|\le C$ for the original $H_\omega$, the size reduction process indicated above creates a new background operator $T_0$ with uncontrollable norm). Our next step in the proof will require that the background operator is bounded in norm by a constant, say by $1/2$. We achieve this by means of the following transformation.

\begin{prop}\label{lem:prel5}
Let $B_{1,2}$ be a pair of Hermitian $L\times L$ matrices with $\|B_1\|\le 1$. Consider the matrices
\[B=B_1+B_2,\quad \hat B=\left(B_1-a I_L\right)^{-1}+
\left(B_2+a I_L\right)^{-1}\] where  $a\in\R$.
 Then there exists an integer $a\in [-L-3,-3]\cup [3,L+3]$ (which depends on $B_2$ but not on $B_1$) and $\epsilon_0>0$ (which depends only on $L$) such that for any $\epsilon<\epsilon_0$
\begin{align}
\max\left(\left\|\left(B_1-a I_L\right)^{-1}\right\|,
\left\|\left(B_2+a I_L\right)^{-1}\right\|\right)\ \le \frac{1}{2} \,; \label{eq:red2}
\end{align}
\begin{align}
{\mathcal C}_{\epsilon/(225L^4)}\left(\hat
B\right)\,\le\,{\mathcal C}_{\epsilon}\left(B\right)\,\le\,
{\mathcal C}_{7 L^2\epsilon}\left(\hat B\right)\,.\label{eq:red}
\end{align}
\end{prop}
 We will apply this proposition to the operator $H_\omega-E$ by choosing $B_2=H_0-E$, $B_1=H_\omega-H_0$. By the hypothesis of Theorem \ref{thm:main}, the assumptions of Proposition  \ref{lem:prel5} are satisfied, with $L=kN$.  Combining this observation with the size reduction, we obtain the second intermediate result.
 \begin{lemma}\label{lem:hatH}
Assume {\bf (A)}. Let $\hat H_\omega$ be an operator  acting on $\ell^2({\mathcal V}; \C^k)$ as
 \be\label{eq:hatHd} (\hat H_\omega\psi)(x)=(H_0\psi)(x)+(A_\omega(x)-a)^{-1}\psi(x).\ee
 Suppose that $\|H_0\|\le 1/2$.
 If there  exist $\epsilon_0>0$ and  $b>0$ (which depend on $k,m,\alpha$) so that for all integers $a\in [-km-3,-3]\cup [3,km+3]$ and all $N\le 2m$ the bound
 \[
 \pr\left({\mathcal C}_{7k^2m^2\epsilon}(\hat H_\omega)\ge m\right)\le  b|\ln\epsilon|^m\epsilon^{\alpha}
 \]
 holds uniformly in $H_0$ and $\epsilon<\epsilon_0$, then the second assertion of Theorem \ref{thm:main} holds.
 \end{lemma}
\paragraph{\it Reduction to the determinant}

Suppose that ${\mathcal C}_{7k^2m^2\epsilon}(\hat H_\omega)\ge m$. Since  by construction $\|\hat H_\omega\|\le 1$,  the operator $\hat H_\omega$ can have no more than $kN-m$ large eigenvalues, and each of these can have an absolute value no larger than $1$.  As a result, we obtain the bound:
\be\label{eq:detest}|\det \hat H_\omega|\le (7k^2m^2\epsilon)^m .\ee
We may now employ the following lemma to calculate the probability of the aforementioned bound on the determinant. Its proof can be found in Section \ref{sec:proofs}.

\begin{lemma}\label{lem:main}
 Assume {\bf (A)}. Let $\hat H_\omega$ be as in \eqref{eq:hatHd}, and let $\hat {\mathcal E}_\delta$ be the event
\[\hat {\mathcal E}_\delta = \{\omega\in\Omega: \ \det(\hat H_\omega)\le \delta\}.\]
Let
\[\delta_0:= \exp\left(2K\alpha^{1+1/N)}\right).\]
Then for any $\delta\in[0,\delta_0]$ we have
\be\label{eq:main'}{\mathbb P}(\hat {\mathcal E}_\delta)\le  (2K\alpha)^N\ln^N(\delta^{-1})\delta^{\alpha}.\ee
\end{lemma}
Using this result in conjunction with \eqref{eq:detest} we obtain the there exist $\epsilon_0>0$ and  $b>0$ that depend on $k,m,\alpha$ so that
\be\label{eq:bndwdet}
\pr \left({\mathcal C}_{7k^2m^2\epsilon}(\hat H_\omega)\ge m \right)\le  b|\ln\epsilon|^m\epsilon^{\alpha},\ee
for $N\le m$ and for $\epsilon<\epsilon_0$.
The combination of Lemma \ref{lem:hatH} and \eqref{eq:bndwdet} yields \eqref{eq:mainbnd}.

\vspace{3mm}

\paragraph{\it Improvement on the Wegner bound}
We want to improve on this bound for the special case that $m=1$. In this case we need to verify the (improved) input for Lemma \ref{lem:interm} for $N=1$ and $N=2$.  In the former case, the bound \eqref{eq:mainbnd1} follows  from {\bf(A)} and Proposition \ref{lem:prel5} (where we choose $B_2=H_0-E$, $B_1=H_\omega-H_0$). So for the rest of the argument we will assume that $N=2$.

Let ${\mathcal E}_\epsilon,\,{\mathcal S}_\epsilon$ be the events
\begin{eqnarray*} {\mathcal E}_\epsilon &=& \{\omega:\ {\mathcal C}_{\epsilon}(H_\omega-E)\ge1\};\\
 {\mathcal S}_\epsilon &=& \{\omega:\ {\mathcal C}_{\epsilon^{2/3}}(H_\omega-E)\ge2\}.\end{eqnarray*}
 We first observe that it follows from \eqref{eq:mainbnd} (which we already established earlier) that
\[\pr({\mathcal E}_\epsilon \cap {\mathcal S}_\epsilon)\le \pr({\mathcal S}_\epsilon)\le C \left|\ln(\epsilon^{2\alpha/3})\epsilon^{2\alpha/3}\right|^{2}\le C\epsilon^\alpha\] for $\epsilon$ sufficiently small. Therefore, to get \eqref{eq:mainbnd1} it suffices to show that
$\pr({\mathcal E}_\epsilon\smallsetminus {\mathcal S}_\epsilon)\le C\epsilon^\alpha$. To this end, suppose that $\omega\in  {\mathcal E}_\epsilon\smallsetminus {\mathcal S}_\epsilon$.
Then
\be\label{eq:posshift}
(H_\omega-E+\epsilon)^{-1}+2\epsilon^{-2/3}>0;\ \ \left\|(H_\omega-E+\epsilon)^{-1}+2\epsilon^{-2/3}\right\|\ge \frac{1}{2\epsilon}.\ee
If ${\mathcal V}=\{x,y\}$, let us denote by $P_x$ ($P_y$) the rank $k$ projection onto the site $x$ (accordingly $y$).
The positivity of the left-hand side can be exploited by means of Lemma \ref{lem:projbnd} below with choices  $P_1=P_x$, $P_2=P_y$.

\begin{lemma}\label{lem:projbnd}
Let $A>0$, and let $P_{1,2}$ be orthogonal projections that satisfy $P_1P_2=0$. Let $P=P_1+P_2$. Then we have
\begin{equation}\label{eq:PAP}
\|PAP\|\le 2\max(\|P_1AP_1\|,\|P_2 A P_2\|)\,.
\end{equation}
\end{lemma}

Using \eqref{eq:posshift} and \eqref{eq:PAP}, we infer that  $\omega\in {\mathcal R}_{\epsilon}$ (and thus ${\mathcal E}_\epsilon\smallsetminus {\mathcal S}_\epsilon\subset {\mathcal R}_{\epsilon}$), where
\[{\mathcal R}_{\epsilon}=\left\{\omega:\ \max_{i=x,y}\left(\left\|P_i (H_\omega-E+\epsilon)^{-1} P_i\right\|+2\epsilon^{-2/3} \right) \geq \frac{1}{4\epsilon} \right\}.\]
But
\begin{multline*}P_x (H_\omega-E+\epsilon)^{-1} P_x =\Big((H_\omega-E+\epsilon)/P_x (H_\omega-E+\epsilon) P_x\Big)^{-1}\\ =(A_\omega(x)+J)^{-1}\end{multline*}
by the block inversion formula. Here the operator $J$ depends on $A_\omega(y)$ but not on $A_\omega(x)$. Hence we can deduce from  {\bf(A)} and Proposition \ref{lem:prel5} that
\[\pr \left(\left\|P_x (H_\omega-E+\epsilon)^{-1} P_x\right\|\ge \frac{1}{5\epsilon}\right)\le\tilde C \epsilon,\]
with $\tilde  C$ that depends on $k,\alpha$ but not on $\epsilon$. The same bound holds with $P_x$ replaced by $P_y$. Hence we infer that for $\epsilon$ small enough
\[\pr({\mathcal R}_\epsilon)\le C\epsilon^\alpha, \]
and the result follows. \hfill \qed

\begin{proof}[Proof of Lemma \ref{lem:projbnd}]
Let $A_1=P_1AP_1$, $A_2=P_2AP_2$, and $A_{12}=P_1AP_2$. Then  by Schur complement condition for positive definiteness
\[A_2>0\,;\quad A_1\ge A_{12}A_2^{-1}A_{21}.\]
Since $A_2$ is positive, $A_2^{-1}\ge 1/\|A_2\| $, hence
\[A_1\ge A_{12}A_{21}/\|A_2\|\,,\] and so
\[\|A_1\|\|A_2\|\ge \|A_{12}\|^2\,,\] where in the last step we have used $A_{12}=A_{12}^*$. Since
\[\|PAP\|\le \max(\|A_1\|+\|A_{12}\|, \|A_2\|+\|A_{21}\|)\,, \]the result follows.
\end{proof}


\section{Proof of Theorems \ref{thm:weakMint} and  \ref{thm:weakMin}}\label{sec:appl}
\subsection{Proof of Theorem  \ref{thm:weakMint}}
Let $a$ be an integer that satisfies $a-b\ge2$, and let $j$ be arbitrary fixed real number. Then if $\epsilon\in[0,1/(2a)]$, the inequality
\be\label{eq:andbnd}\left|\frac{1}{v_x-a}+\frac{1}{j+a}\right|<\epsilon\ee
for $v_x$ has solutions in $I$ only if $0<j+a<1$. Since equation
\[\left|\frac{1}{v_x-a}+\frac{1}{j+a}\right|=\epsilon\] define the pair of points
\[v_x=a+\frac{1}{(j+a)^{-1}\pm\epsilon},\]
the set of $v_x$ for which \eqref{eq:andbnd} holds is the interval $\hat I$ of length
\[|\hat I|=\frac{1}{(j+a)^{-1}-\epsilon} - \frac{1}{(j+a)^{-1}+\epsilon}=\frac{2\epsilon}{(j+a)^{-2}-\epsilon^2}< 4\epsilon, \]
where in the last step we used $0<a+j<1$, $\epsilon<1/(2a)<1/4$. Hence
\[\pr\left(\left|\frac{1}{v_x-a}+\frac{1}{j+a}\right|<\epsilon\right)\le C \epsilon^\beta\] by the hypothesis of the theorem. \hfill \qed
\subsection{Proof of Theorem  \ref{thm:weakMin}}
We first establish the bounds
\begin{eqnarray}\label{eq:fbnd4}\pr(|\det (\sigma_\omega(x)+J)|\le \epsilon)&\le& C\epsilon^{\beta};\\
\label{eq:fbnd5} \pr({\mathcal C}_{\epsilon}(\sigma_\omega(x)+J)\neq0)&\le& C\epsilon^{\beta}
\end{eqnarray} for $\epsilon\in[0,1]$. Indeed, note that $\det (\sigma_\omega(x)+J)=c^2-(u_x-a)^2-(v_x-b)^2$ with some constants $a,b,c$ originating from $J$. Therefore the set $|\det (\sigma_\omega(x)+J)|\le \epsilon$ is an intersection $I$ of the disc $O$ with the annulus centered at $a,b$ and with radii $R_-=\sqrt{\max(c^2-\epsilon,0)}$,  $R_+=\sqrt{c^2+\epsilon}$. The area of this set therefore cannot exceed $\pi(R_+^2-R_-^2)\le2\pi\epsilon$ and \eqref{eq:fbnd4}  follows.  To establish \eqref{eq:fbnd5}, we note that
\[{\mathbb P}({\mathcal C}_{\epsilon}(\sigma_\omega(x)+J)\neq0)={\mathbb P}(\|(\sigma_\omega(x)+J)^{-1}\|\ge1/\epsilon).\]
The value of $\|(\sigma_\omega(x)+J)^{-1}\|$ however can be evaluated explicitly and is given by
\[\|(\sigma_\omega(x)+J)^{-1}\|=\left||c|-\sqrt{(u_x-a)^2+(v_x-b)^2}\right|^{-1},\]
with the same constants $a,b,c$ as before. Hence the set of the points in $O$ that satisfy $\|(\sigma_\omega(x)+J)^{-1}\|\ge1/\epsilon$ is an intersection $\hat I$ of the disc $O$ with the annulus centered at $a,b$ and with radii $R_-=\left||c|-\epsilon\right|$,  $R_+=|c|+\epsilon$. The area of $\hat I$ cannot exceed the circumference of the unit circle times the maximal thickness $2\epsilon$ of the annulus, so $|\hat I|\le 4\pi\epsilon$, and \eqref{eq:fbnd5} follows.

The assertion of the theorem  follows now from  Lemma \ref{lem:Either} below (whose proof can be found in Section \ref{sec:proofs}) and bounds \eqref{eq:fbnd4} - \eqref{eq:fbnd5}.

\hfill \qed
\begin{lemma}\label{lem:Either}
Let $A$ and $J$ be Hermitian $k\times k$ matrices, and let $a$ be the real number that satisfies $|a|\ge2$. If $\|A\|\le 1$ and
\[\left|\det \left((A-a)^{-1}+(J+a)^{-1}\right)\right|\le \frac{1}{16|a|}\left\{\frac{|a|-1}{2(|a|+1)^2}\right\}^{k-1}\]
then we have
\begin{eqnarray*} && \left(2(|a|+1)^2\right)^{-k}\left|\det \left(A+J \right)\right|\le  \left|\det \left((A-a)^{-1}+(J+a)^{-1} \right)\right|
\end{eqnarray*}
or
\[
 \frac{1}{16|a|}\left\{2(|a|+1)^2\left\|\left(A+J\right)^{-1}\right\|\right\} ^{1-k}\le \left|\det \left((A-a)^{-1}+(J+a)^{-1}\right) \right|.
\]
\end{lemma}




%
%

%

\section{Proofs}\label{sec:proofs}


\begin{proof}[Proof of Proposition \ref{lem:prel5}]
If $|a|\ge 3$, then since $||B_1||\leq 1$, we have:
\be\label{eq:bndB1} \nu:=\left\|\left(B_1-a\right)^{-1}\right\|,\quad (L+4)^{-1}\le\nu\le 1/2\,.\ee
Since  $B_{2}$ is $L\times L$, it has at most $L$ distinct eigenvalues. On the other hand, for every set $S$ of real numbers with $|S|=L$
there exists an integer $a\in [-L-3,-3]\cup [3,L+3]$ so that ${\rm
dist}(S,-a)\ge 2$, hence we can choose $a$ that satisfies \eqref{eq:red2}.

With this choice of $a$, consider  the block matrix $W$ of the form
\[W=\left[\begin{array}{cc}
\left(B_2+aI_n\right)^{-1} & \left(B_1-aI_n\right)^{-1} \\
\left(B_1-aI_n\right)^{-1} & -\left(B_1-aI_n\right)^{-1}\\
\end{array}\right]\,.\] Note that the Schur complements to the upper
and lower diagonal blocks are
\begin{eqnarray*}-\, W/W_{11}&=&\left(B_1-a\right)^{-1}\,+\,
\left(B_1-a\right)^{-1}\left(B_2+a\right)\left(B_1-a\right)^{-1}\,;
\\
W/W_{22}&=&\left(B_2+a\right)^{-1}\,+\,\left(B_1-a\right)^{-1}\,.
\end{eqnarray*}
Let
\be\label{eq:Top}T=\left(B_1-a\right)^{-1}/\nu,\ee
 where $\nu$ is given by \eqref{eq:bndB1}.

In what follows, we will need two lemmas:
\begin{lemma}\label{lem:prel1}
\hspace{0.1 mm}
\begin{enumerate}
\item
Let $D=D^*, \tilde D=\tilde D^* \in M_{n,n}$. Then
\[{\mathcal C}_\epsilon(D)\le {\mathcal C}_{2\epsilon}(\tilde D)\,,\] provided
$\|D-\tilde D\|\le \epsilon$.
\item
${\mathcal C}_{\epsilon}\left(A\right)\,\le\,{\mathcal C}_{\epsilon}\left(BAB\right)$
whenever \[ A=A^*, B=B^*,  \|B\|\,\le\,1\] \label{eq:ceprod}
\end{enumerate}
\end{lemma}

\begin{lemma}\label{lem:prel3}
Suppose $D=D^*\in M_{n,n}$ is of the form
\[D=\left[\begin{array}{cc}
A & V \\
 V^* &B\\
\end{array}\right]\,,\] with $A\in M_{k,k},\ B\in M_{m,m}$,
 $\|V\|\le 1/2$, and
\begin{equation}\label{eq:conde}{\mathcal C}_{2\epsilon}(B)=0\,.\end{equation}
Then
\begin{equation}\label{eq:1}
{\mathcal C}_{\epsilon}\left(D/B\right)\le {\mathcal C}_{\epsilon}\left(D\right)\,;
\end{equation}
\begin{equation}\label{eq:2}
{\mathcal C}_{\epsilon}\left(D\right)\le {\mathcal C}_{\beta\epsilon}\left(D/B\right) \,,
\end{equation} with
$\beta=2(\|B^{-1}\|+1)^2$.
\end{lemma}

Armed with these results, we can infer that
\begin{multline}\label{eq:nor}
{\mathcal C}_{\epsilon}\left(B\right)\,=\,{\mathcal C}_{\nu^2\epsilon}\left(\nu^2B\right)\,\le\, {\mathcal C}_{\nu^2\epsilon}\left(\nu^2TBT\right)\, =\,
{\mathcal C}_{\nu^2\epsilon}\left(W/W_{11}\right)\\\le\,
{\mathcal C}_{\nu^2\epsilon}\left(W\right)\,\le\,
{\mathcal C}_{\beta\nu^2\epsilon}\left(W/W_{22}\right)\,=\,
{\mathcal C}_{\beta\nu^2\epsilon}\left(\hat B\right) \,,
\end{multline}
where in the second step we have used Lemma \ref{lem:prel1} and in the remaining steps we have used  Lemma
\ref{lem:prel3}. Here
\[\beta\,=\,
2\left(\left\|\left(W_{22}\right)^{-1}\right\|\,+\,1\right)^2\,\le\,2\left(L+5\right)^2\,\le\,
25L^2\] for $L \geq 2$.  (It is straightforward to check that the relation ${\mathcal C}_{\epsilon}\left( B\right)\le {\mathcal C}_{\beta\nu^2\epsilon}\left(\hat B\right)$  holds for $L=1$.)  Plugging in the upper bounds for $\nu,\beta$ we get the second inequality in \eqref{eq:red}:
\[{\mathcal C}_{\epsilon}\left(B\right)\,\le\,{\mathcal C}_{7 L^2\epsilon}\left(\hat B\right).\]

On the other hand, let
\[U=\kappa \left[\begin{array}{cc}
B_2+aI_n & B_1-aI_n \\
B_1-aI_n & -B_1+aI_n\\
\end{array}\right];\quad \kappa=\frac{1}{2\|B_1-a\|}\,.\] Then
\begin{eqnarray*} U/U_{22}&=&\kappa B\,;
\\
-U/U_{11}&=&\kappa\left(B_1-a\,+\,
\left(B_1-a\right)\left(B_2+a\right)^{-1}\left(B_1-a\right)\right)\,.
\end{eqnarray*}
Similarly to \eqref{eq:nor}, we obtain
\begin{multline*}
 {\mathcal C}_{\nu^{2}\epsilon/(25L^2)}\left(\hat
B\right)\, =\,{\mathcal C}_{\epsilon/(25L^2)}\left(\kappa^{-1}T\left(U/U_{11}\right)T\right)\\ \le\,
{\mathcal C}_{\kappa\epsilon/(25L^2)}\left(U/U_{11}\right)\, \le\,
{\mathcal C}_{\kappa\epsilon/(25L^2)}\left(U\right)\\ \le\,
{\mathcal C}_{\kappa\beta\epsilon/(25L^2)}\left(U/U_{22}\right)\, \le\, {\mathcal C}_{\epsilon}(B)\,,
\end{multline*} with $T$ given by \eqref{eq:Top} and
\[\beta\,=\,
2\left(\left\|\left(U_{22}\right)^{-1}\right\|\,+\,1\right)^2\,\le\,2\left(\|B_1-a\|+1\right)^2\le 25L^2\,.\]
Since 
\[\frac{\nu^{2}\epsilon}{25L^2}\ge \frac{\epsilon}{25L^2(L+4)^2}\ge \frac{\epsilon}{225L^4},\] the first inequality in \eqref{eq:red} follows.

\end{proof}


\begin{proof}[Proof of Lemma \ref{lem:prel1}]
For the first part, we use the Weyl's theorem, cf. Theorem 4.3.1 in \cite{HJ}, which states that if
\[\sigma(A)=\{\lambda_i(A)\}_{i=1}^n\,,\ \ 
\sigma(B)=\{\lambda_i(B)\}_{i=1}^n\,, \ \ 
\sigma(A+B)=\{\lambda_i(A+B)\}_{i=1}^n\] for Hermitian $A,\ B$, with
the eigenvalues arranged in increasing order, then
\[\lambda_k(A)+\lambda_1(B)\le \lambda_k(A+B) \le
\lambda_k(A)+\lambda_n(B)\,,\quad k=1,...,n\,.\] Therefore, every
number $\lambda_k(A+B)$ which lies in the interval
$[-\epsilon,\epsilon]$  can be matched with
$\lambda_k(A)\in[-2\epsilon,2\epsilon]$, provided that
$\|B\|\le\epsilon$.

\vspace{.6cm}
\noindent
For the second part, observe that there exists a Hermitian matrix $\hat A$ such that
\begin{enumerate}
\item $\|\hat A- A\|\le \epsilon$;
\item
${\rm nul}(\hat A)={\mathcal C}_{\epsilon} (A)$;
\item $\hat A$ has no non-zero
eigenvalues in the interval $(-\epsilon,\epsilon)$.
\end{enumerate}
Then Sylvester's law of inertia implies that ${\rm nul}(\hat A)\le {\rm nul}(B\hat A B)$ (with equality in the case of nonsingular $B$). Since $\|BAB-B\hat AB\|\le \epsilon$ we can use Weyl's theorem again to conclude that \[{\mathcal C}_{\epsilon} (BAB)\ge {\rm nul} (B\hat AB)\ge{\rm nul}(\hat A)={\mathcal C}_{\epsilon} (A).\]
\end{proof}


\begin{proof}[Proof of Lemma \ref{lem:prel3}]
The relation \eqref{eq:1} follows from the interlacing theorem for
inverses of Hermitian matrices---see Lemma 2.3 in \cite{Z},
which is itself a simple consequence of the the Schur complement
formula and Cauchy interlacing theorem for Hermitian matrices.

To prove \eqref{eq:2} note that there exists a matrix \[\hat
D:=\left[\begin{array}{cc}
\hat A & \hat V \\
\hat  V^* &\hat B\\
\end{array}\right]\]  such that
\begin{enumerate}
\item\label{it:1} $\|\hat D- D\|\le \epsilon$;
\item
${\rm nul}\,\hat D={\mathcal C}_{\epsilon} (D)$;
\item $\hat D$ has no non-zero
eigenvalues in the interval $[-\epsilon,\epsilon]$.
\end{enumerate}
where ${\rm nul}\,\hat D$ is the multiplicity of the zero eigenvalue of $\hat D$, and equals zero if this eigenvalue is absent.  One can readily prove the existence of $\hat D$ by diagonalizing $D$ and replacing all eigenvalues less than or equal to $\epsilon$ with zeros.
Using the Haynsworth inertia additivity formula, we get
\[{\rm nul}\, \hat D={\rm nul}\, \hat B + {\rm nul}\, \left(\hat
D/ \hat B\right)\,.\] Observe that the condition \eqref{it:1} above
implies $\|\hat B- B\|\le \epsilon$. We can therefore
infer from \eqref{eq:conde} and Lemma \ref{lem:prel1} that
\begin{equation}\label{eq:hatB}
{\mathcal C}_{\epsilon}( \hat B )=0\,.
\end{equation}
As a result we obtain the equality
\begin{equation}\label{eq:hateql}{\mathcal C}_{\epsilon} (D)={\rm nul}\, \hat D={\rm nul}\, \left(\hat
D/ \hat B\right)\,.
\end{equation}
Note now that
\begin{align*}
&\|\hat V\hat B^{-1}\hat V^*-V B^{-1}V^*\|\ \notag\\
&\leq \  \|(\hat V-V)\|\cdot \|\hat B^{-1}\hat V^*\| + \ \|V\|\cdot\|\hat B^{-1}-
B^{-1}\|\cdot \|\hat V^*\|\ \notag\\
&\hspace{6 mm} + \ \|V\hat B^{-1}\|\cdot \|(\hat V^*-V^*)\| \notag\\
&\leq \epsilon \|\hat B^{-1}\| \left(\frac{1}{2}+\epsilon\right)
+ \frac{1}{2} \|\hat B^{-1} - B^{-1}\| \left(\frac{1}{2}+\epsilon\right)
+ \frac{1}{2} \|\hat B^{-1}\| \epsilon \notag\\
&= \left(\epsilon + \epsilon^2\right) \|\hat B^{-1}\| + \left(\frac{1}{2} \epsilon
+ \frac{1}{4}\right) \|\hat B^{-1} - B^{-1}\| \notag\\
&\leq \frac{3}{2}\epsilon \|\hat B^{-1}\|\ + \frac{1}{2}\|\hat B^{-1}- B^{-1}\| \,.
\end{align*}
(We have assumed that $\epsilon \leq \frac{1}{2}$.)  Using the first resolvent identity, we get the bound
\begin{equation}\label{eq:diffr}\|\hat B^{-1}- B^{-1}\| \ = \ \|\hat
B^{-1}(\hat B- B)B^{-1}\| \ \le \ \epsilon \|B^{-1}\|\cdot \|\hat
B^{-1}\|\ \,,
\end{equation}
which in turn implies the estimate
\[\|\hat B^{-1}\| \  \le \|B^{-1}\| \ + \ \epsilon
\|B^{-1}\|\cdot \|\hat B^{-1}\| \ \le \ 2\|B^{-1}\|\,, \] where we
have used \eqref{eq:hatB} in the last step. Inserting the last
inequality into the right-hand side of \eqref{eq:diffr}, we finally obtain
\[ \|\hat B^{-1}- B^{-1}\| \  \le \ 2
\epsilon \|B^{-1}\|^2\,.\]
As a result, we arrive at
\[\|\hat V\hat B^{-1}\hat V^*-V B^{-1}V^*\|\ \le \ \frac{3}{2}\epsilon
\|B^{-1}\|\ + \ \epsilon\|B^{-1}\|^2\,,\]
hence
\[\left\|\hat D/\hat B - D/ B\right\|\ \le \ \epsilon(1+\frac{3}{2}
\|B^{-1}\|\ + \ \|B^{-1}\|^2) \ <\ \epsilon(\|B^{-1}\|+1)^2\ =:\
\frac{\epsilon\beta}{2}\,.\]
Consequently,  we get
\[{\mathcal C}_{\epsilon}\left(D\right)\ = \ {\rm nul}\, \left(\hat
D/ \hat B\right) \ \le \ {\mathcal C}_{\frac{\beta\epsilon}{2}}\left(\hat D/
\hat B\right)\ \le \  {\mathcal C}_{\beta\epsilon}\left(D/ B\right)\,,
\]
where we have used \eqref{eq:hateql} in the first step and Lemma
\ref{lem:prel1} in the last one.
\end{proof}


\begin{proof}[Proof of Lemma \ref{lem:main}]
We use induction in $N$. For $N=1$ the result follows from {\bf(A)}. Suppose that \eqref{eq:main'} holds for $|{\mathcal V}|=N$. We wish to establish the induction step, i.e. \eqref{eq:main'} for $|{\mathcal V}|=N+1$. We can  evaluate  $\det\hat H_\omega$ using the Schur determinant formula. Namely, for $x\in \mathcal V$ let us denote by $\hat H^{(x)}_{\omega}$ denote the restriction of $\hat H_\omega$ to the site $x$. Then
\[\det\hat H_\omega=\det \hat H^{(x)}_{\omega}\det(\hat H_\omega/\hat H^{(x)}_{\omega})\] by Schur's determinant formula. Both determinants on the right-hand side are random, but the first one depends only on randomness associated with $A_\omega(x)$, a fact which we will exploit momentarily.  We note now that the Schur complement $\hat H_\omega/\hat H^{(x)}_{\omega}$ is by itself also of the form \eqref{eq:hatHd}  (with $\mathcal V$ replaced by ${\mathcal V}\setminus \{x\}$). Note that the  $H_0$ term in $\hat H_\omega/\hat H^{(x)}_{\omega}$ might depend on $A_\omega(x)$, but not on other random variables $\{A_\omega(y)\}$. By the induction hypothesis, we have
\be\label{eq:indhyp}{\mathbb P}\left(|\det(\hat H_\omega/\hat H^{(x)}_{\omega})|\le r\right)\le (2 K \alpha)^N\ln^N(r^{-1})r^{\alpha},\quad r\in[0,1].\ee

Let $S:=\{\omega:\ |\det\hat H_\omega|
\le \epsilon\}$, and let
\[F_\omega=|\det \hat H^{(x)}_{\omega}|,\quad G_\omega=|\det(\hat H_\omega/\hat H^{(x)}_{\omega})|.\]
We set $Q:=\{\omega:\ \min(F_\omega,G_\omega)\le \epsilon\}$, then by Assumption {\bf (A)} and the induction hypothesis
\be\label{eq:Q}\pr(Q)\le \left(K+ (2 K \alpha)^N\ln^N(\epsilon^{-1})\right)\epsilon^{\alpha}.
\ee
On the other hand, we have
\[\chi(S\smallsetminus Q)= \int_\epsilon^{1}\chi(sG_\omega\le\epsilon)\delta(F_\omega-s)ds.\]

Taking expectations on both sides and using \eqref{eq:indhyp}, we obtain
\begin{align}\label{eq:exp1}
	\E \chi(S\smallsetminus Q) &\le \E\int_\epsilon^{1} ds\,\delta(F_\omega-s)
		\E\left(\chi(sG_\omega\le\epsilon)\,\Big|\, A_\omega(x) \right) \notag\\
	& \le (2 K \alpha)^N\epsilon^{\alpha}\E\int_\epsilon^{1}
		\frac{\ln^N(\frac{s}{\epsilon}) \delta(F_\omega-s)}{s^{\alpha}}ds \notag \\
	& = (2 K \alpha)^N \epsilon^{\alpha}\E \,\frac{\ln^N(\epsilon^{-1}F_\omega)\chi(1>F_\omega>\epsilon)}{(F_\omega)^{\alpha}} \notag \\
	& \leq (2 K \alpha)^N \epsilon^{\alpha} \ln^N(\epsilon^{-1}) \E \,\frac{\chi(1>F_\omega>\epsilon)}{(F_\omega)^{\alpha}} .
\end{align}
  Using now  {\bf(A)} and the layer cake representation, we get
\begin{align}\label{eq:exp2}
	\E \,\frac{\chi(1>F_\omega>\epsilon)}{(F_\omega)^{\alpha}}
	& = \int^{\epsilon^{-\alpha}}_{1} \pr \left( (F_\omega)^{-\alpha}\ge t \right) d t
	\le  K \int^{\epsilon^{-\alpha}}_{1}  \frac{1}{t} d t  \notag\\ & = K\alpha \ln(\epsilon^{-1})
\end{align}
Combination of \eqref{eq:Q}, \eqref{eq:exp1}, and \eqref{eq:exp2} yields the induction step. \\
\end{proof}

\begin{proof}[Proof of Lemma \ref{lem:Either}]
We have
\begin{multline*}\left|\det \left(A+J \right)^{-1}\right|\\ =\left|\det (A-a)^{-1}\right|\,\left|\det(J+a)^{-1}\right|\,\left|\det \left((A-a)^{-1}+(J+a)^{-1} \right)^{-1}\right|.\end{multline*}
Suppose first that ${\mathcal C}_{16|a|}\left(\left((A-a)^{-1}+(J+a)^{-1}\right)^{-1}\right)=0$. According to \eqref{eq:ceprod}
\begin{multline*}{\mathcal C}_{16/|a|}\left((A-a)^{-1}\left((A-a)^{-1}+(J+a)^{-1} \right)^{-1}(A-a)^{-1}\right)\\ \le {\mathcal C}_{16/|a|}\left((|a|/2)^{-2}\left((A-a)^{-1}+(J+a)^{-1} \right)^{-1}\right)\\ ={\mathcal C}_{4|a|}\left(\left((A-a)^{-1}+(J+a)^{-1} \right)^{-1}\right)=0,\end{multline*}
where we have used $\|A-a\|\ge|a|/2$.
Since \[\|(A-a)^{-1}\|\le(|a|-1)^{-1},\] we can use \eqref{eq:Wood'} to decompose
\begin{multline*}\det \left(A+J \right)^{-1} = \\ \det \left((A-a)^{-1}\left((A-a)^{-1}+(J+a)^{-1} \right)^{-1}(A-a)^{-1}\right) \det \left(H-I\right),\end{multline*}
with
\[H=\left((A-a)^{-1}+(J+a)^{-1}\right)(A-a).\]
It follows that $\|H\|\le1/2$, and consequently $|\det \left(H-I\right)|\ge 2^k$. On the other hand, $\left|\det(A-a)\right|\le (|a|+1)^k$, and we can  conclude that
\be\label{eq:case1} \left|\det \left(A+J \right)^{-1}\right|\ge \left(2(|a|+1)^2\right)^{-k} \left|\det \left((A-a)^{-1}+(J+a)^{-1} \right)^{-1}\right|,\ee
whenever
${\mathcal C}_{16|a|}\left(\left((A-a)^{-1}+(J+a)^{-1} \right)^{-1}\right)=0$.

On the other hand, if ${\mathcal C}_{16|a|}\left(\left((A-a)^{-1}+(J+a)^{-1} \right)^{-1}\right)\neq 0$, then
\begin{multline*}\left\|\left((A-a)^{-1}+(J+a)^{-1} \right)^{-1}\right\|^{k-1}\\ \ge \frac{1}{16|a|}\left|\det \left((A-a)^{-1}+(J+a)^{-1}\right)^{-1}\right|.\end{multline*}
Hence
\begin{multline*}\left\|(A-a)^{-1}\left((A-a)^{-1}+(J+a)^{-1} \right)^{-1}(A-a)^{-1}\right\|\\ \ge (|a|+1)^{-2}\left\|\left((A-a)^{-1}+(J+a)^{-1} \right)^{-1}\right\|\\ \ge(|a|+1)^{-2}\left\{\frac{1}{16|a|}\left|\det \left((A-a)^{-1}+(J+a)^{-1} \right)^{-1}\right|\right\}^{1/(k-1)}.\end{multline*}
Using \eqref{eq:Wood'}, we conclude that
\begin{multline}\label{eq:case2}\left\|\left(A+J \right)^{-1}\right\|\\ \ge(|a|+1)^{-2}\left\{\frac{1}{16|a|}\left|\det \left((A-a)^{-1}+(J+a)^{-1} \right)^{-1}\right|\right\}^{1/(k-1)}\\ -(|a|-1)^{-1}\\ \ge\frac{1}{2}(|a|+1)^{-2}\left\{\frac{1}{16|a|}\left|\det \left((A-a)^{-1}+(J+a)^{-1} \right)^{-1}\right|\right\}^{1/(k-1)},\end{multline}
if ${\mathcal C}_{16|a|}\left(\left((A-a)^{-1}+(J+a)^{-1} \right)^{-1}\right)\neq 0$,
whenever  \eqref{eq:conda} holds.

Combining \eqref{eq:case1} and \eqref{eq:case2} we establish the assertion.
\end{proof}


\subsection*{Acknowledgments:}   We are grateful to Sasha Sodin for valuable comments pertaining to the main result. This work was supported in part by NSF grant DMS--1210982.

\end{document}